\documentclass[11pt]{article}

\usepackage[margin=1in]{geometry}

\usepackage{amsmath}
\usepackage{amssymb}
\usepackage{amsthm}
\usepackage{mathtools}

\usepackage{booktabs}
\usepackage{array}
\usepackage{multirow}

\usepackage{tikz}
\usetikzlibrary{positioning, arrows.meta, shapes.geometric, fit, calc}

\usepackage[colorlinks=true,linkcolor=blue,citecolor=blue,urlcolor=blue]{hyperref}

\usepackage{natbib}

\usepackage{enumitem}
\usepackage{microtype}

\usepackage{listings}
\newcommand{\govern}{\mathcal{G}}
\newcommand{\policy}{\mathcal{P}}
\newcommand{\ledger}{\mathcal{L}}
\newcommand{\runtime}{\mathcal{R}}

\newcommand{\realize}{\mathit{realize}}
\newcommand{\record}{\mathit{record}}

\newtheorem{theorem}{Theorem}[section]

\newtheorem{property}{Property}
\theoremstyle{definition}
\newtheorem{definition}[theorem]{Definition}
\theoremstyle{remark}

\begin{document}

\title{Intent-Driven Computing: A Computational Model\\for Governed Autonomous Systems}

\author{Alan L. McCann\\
\textit{Mashin, Inc.}\\
\texttt{research@mashin.live}}

\date{May 2026}

\maketitle

\begin{abstract}
Programming languages assume programs directly execute effects.
When autonomous systems generate behavior dynamically, this assumption
becomes problematic: there is no structural mediation point between
deciding to act and acting. We define \emph{intent-driven computing}:
a programming model where programs produce \emph{intents} (finite data
values describing proposed actions) rather than directly executing
effects. A governed runtime examines each intent against a decidable
policy language, records every decision in a tamper-evident ledger,
and only then realizes the effect. The language provides no alternative
path to effects. The model does not decide arbitrary behavioral
properties of programs (which Rice's theorem shows is impossible).
Instead, it constrains the language so that all effectful interaction
is reified as finite intent values, shifting governance from the
undecidable domain of program semantics to the decidable domain of
intent data. This yields emergent properties: event sourcing by
construction, governance simulation via intent replay, structural
audit completeness, and improved human comprehensibility. We specify
the model formally, implement it in a concrete language compiling to
the BEAM virtual machine, and verify key properties in Rocq (454
theorems, 36 modules, zero admitted lemmas). Property-based testing
(70,000+ random inputs, zero disagreements) validates that the
implementation matches the specification.
\end{abstract}

\section{Introduction}
\label{sec:intro}

Consider an AI agent deployed to process customer refund requests at a
financial services firm. The agent accesses customer email, a CRM
system, a payment processing API, and internal policy documents. On a
given day it processes 200 requests. The compliance team asks: for each
refund the agent issued, can we prove that the agent was authorized to
issue it, that it checked the policy, that it did not access records it
was not entitled to see, and that the complete decision trail is
independently verifiable?

With common agent frameworks, the answer is generally no unless the
framework is embedded in a separate enforcement substrate. The agent
runs in
Python or JavaScript. It has access to \texttt{requests.post()},
\texttt{subprocess.run()}, or equivalent ambient I/O capabilities.
Guardrails check model outputs but not programmatic effects. API keys
scope access but do not record why access was granted. Logs capture what
happened but cannot prove completeness: there is no guarantee that every
action produced a log entry, and the logs themselves are stored in
infrastructure the auditor must trust.

This paper argues that the problem is not insufficient tooling but an
incorrect abstraction level. Traditional computing conflates two
distinct concerns: the expression of desired actions and the realization
of those actions as effects. When a human writes every line of code,
this conflation is harmless. When autonomous systems generate behavior
dynamically, it becomes structurally problematic.

We propose \emph{intent-driven computing}: a computational model that
separates these concerns by construction. Programs produce
\emph{intents}, not effects. A governed runtime mediates every intent.
The separation is not a policy layer or middleware; it is a property of
the language semantics. Prior work has established the formal
foundations: the structural gap between expressiveness and governance
boundaries~\cite{mccann2026structural}, mechanized proofs of
governance properties~\cite{mccann2026mechanized}, effect-transparent
governance~\cite{mccann2026gcc}, and certified
purity~\cite{mccann2026purity}. This paper contributes the unifying
computational model and its emergent consequences.

\subsection{The Completeness Barrier}

Rice's theorem~\cite{rice1953classes} establishes that for any
non-trivial semantic property $P$ of programs in a Turing-complete
language, the set of programs satisfying $P$ is undecidable. ``This
agent only accesses approved APIs'' is such a property. ``This agent
stays within its refund limit'' is such a property. No behavioral
analysis can be simultaneously sound, complete, and decidable for
these properties over arbitrary programs.

Current AI governance approaches accommodate this barrier by accepting
incompleteness. Content filters~\cite{rebedea2023nemo} examine what
models say, not what programs do. Guardrails~\cite{bai2022constitutional}
wrap individual calls but cannot cover compositions. Runtime
monitors~\cite{leucker2009brief} observe effects after execution, when
damage may already be done. Managed platforms provide scoped permissions
but store audit data in vendor infrastructure, making independent
verification impossible.

We call this the \emph{accountability gap}: the distance between what
governance systems promise and what they can prove. The gap is not an
engineering shortcoming. It is a mathematical consequence of
representing governance as a behavioral property of programs in a
Turing-complete language.

\subsection{The Abstraction Shift}

Intent-driven computing addresses the accountability gap by changing
the abstraction level. Programs in an intent-driven language cannot
cause effects. They can only produce intents: finite, structured data
representing proposed actions. A governance interpreter examines each
intent against a decidable policy language, records the decision in a
tamper-evident ledger, and either realizes the effect or denies the
intent. No alternative effect path exists in the language semantics.

The model does not attempt to decide arbitrary behavioral properties
of programs (which Rice's theorem shows is impossible). Instead, it
constrains the language so that all effectful interaction is reified
as finite intent values. Governance decisions are made over those
values, using a decidable policy language, before realization. This
is a change in the governance domain, not a workaround for
undecidability.

\subsection{Contributions}

\begin{enumerate}[leftmargin=*]
  \item A formally specified programming model, presented as a
    small-step operational calculus with explicit governance mediation
    semantics, in which effectful computation is expressed as
    inspectable intent values and the only operational rule that
    realizes effects must first produce a governance decision and
    ledger event (\S\ref{sec:model}).

  \item A precise account of how this model avoids reducing governance
    to an undecidable behavioral property: enforcement operates on
    finite intent values under a decidable policy language, not on
    program behavior (\S\ref{sec:properties}).

  \item An analysis of how intent-driven computing changes the
    programming experience: writing, reasoning, testing, debugging,
    and evolving autonomous systems (\S\ref{sec:programming}).

  \item A concrete implementation with mechanized proofs (454 Rocq
    theorems, zero admitted), performance measurements, and a case
    study (\S\ref{sec:architecture}, \S\ref{sec:evaluation}).

  \item A comparative analysis situating the model within effect
    systems, capability-based security, event sourcing, monadic I/O,
    and desired-state systems (\S\ref{sec:related}).

  \item A historical argument that intent-driven computing follows the
    pattern of prior abstraction shifts in computing
    (\S\ref{sec:historical}).
\end{enumerate}

\section{The Intent-Driven Computational Model}
\label{sec:model}

\subsection{Definitions}

\begin{definition}[Intent]
\label{def:intent}
An \emph{intent} is a finite, structured data value $i = (\mathit{action},
\mathit{target}, \mathit{params}, \mathit{context})$ representing a
proposed effectful operation. Intents are values, not computations: they
can be inspected, compared, serialized, and decided upon without
executing them.
\end{definition}

\begin{definition}[Effect]
\label{def:effect}
An \emph{effect} is the realization of an intent by the runtime: the
actual execution of the proposed operation (e.g., sending an HTTP
request, writing a file, invoking an external service).
\end{definition}

\begin{definition}[Intent-Driven Program]
\label{def:program}
An \emph{intent-driven program} is a computation that produces a
sequence of intents $[i_1, i_2, \ldots, i_n]$ but cannot directly
execute effects. The program's expressiveness for pure computation is
unrestricted (Turing-complete). Its interaction with the external world
is mediated entirely through intent production.
\end{definition}

\begin{definition}[Governance Interpreter]
\label{def:governance}
A \emph{governance interpreter} $\govern$ is a function that, given a
policy set $\policy$, an intent $i$, and a runtime context $c$, produces
a decision $d \in \{\mathit{allow}, \mathit{deny}, \mathit{escalate}\}$
and a decision record $r$:
\[
  \govern(\policy, i, c) = (d, r)
\]
The governance interpreter is total: every intent receives a decision.
No intent bypasses the interpreter. We require the runtime context $c$
used by governance predicates to be finite and serializable; policies
over non-terminating or externally interactive context are outside the
model. This constraint ensures that governance evaluation terminates
and that context can be recorded for replay.
\end{definition}

\begin{definition}[Intent-Driven Runtime]
\label{def:runtime}
An \emph{intent-driven runtime} $\runtime$ processes a program's intent
sequence as follows. For each intent $i_k$:
\begin{enumerate}
  \item $\govern(\policy, i_k, c_k) = (d_k, r_k)$
  \item $\record(r_k, \ledger)$ --- append $r_k$ to tamper-evident ledger
  \item If $d_k = \mathit{allow}$: $\realize(i_k) = e_k$ --- execute effect
  \item If $d_k = \mathit{deny}$: no effect; program receives denial
  \item If $d_k = \mathit{escalate}$: suspend for external authorization
\end{enumerate}
\end{definition}

\subsection{The Two-Layer Architecture}

Intent-driven programs operate in two layers:

\paragraph{Layer 1: Pure Computation.}
General-purpose computation with no I/O capability. Programs can
perform arithmetic, string processing, data transformation, pattern
matching, control flow, and any other computation expressible in a
Turing-complete language. This layer cannot produce effects because
effect capabilities are structurally absent, not merely restricted.

\paragraph{Layer 2: Governed Effects.}
All interaction with the external world occurs through intent
production. When a program needs to send an email, read a file, invoke
an API, or perform any effectful operation, it produces an intent. The
governance interpreter mediates the intent before any effect occurs.

This separation is enforced by the language semantics, not by runtime
policy. Programs in Layer~1 cannot ``escape'' to ambient I/O because
the language does not provide ambient I/O constructs. The only path to
effects is through intent production, which is mediated by the
governance interpreter.

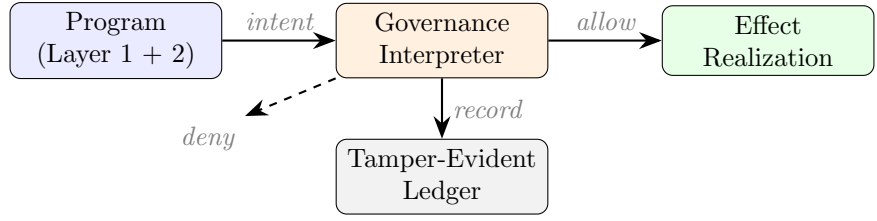
\begin{figure}[t]
\centering
\begin{tikzpicture}[
  node distance=0.8cm and 1.5cm,
  box/.style={draw, rounded corners, minimum width=2.8cm, minimum height=0.8cm, align=center, font=\small},
  arrow/.style={-{Stealth[length=3mm]}, thick},
  label/.style={font=\small\itshape, text=gray}
]
  \node[box, fill=blue!8] (prog) {Program\\(Layer 1 + 2)};
  \node[box, fill=orange!12, right=of prog] (gov) {Governance\\Interpreter};
  \node[box, fill=green!10, right=of gov] (eff) {Effect\\Realization};
  \node[box, fill=gray!10, below=of gov] (ledger) {Tamper-Evident\\Ledger};

  \draw[arrow] (prog) -- node[above, label] {intent} (gov);
  \draw[arrow] (gov) -- node[above, label] {allow} (eff);
  \draw[arrow] (gov) -- node[right, label] {record} (ledger);
  \draw[arrow, dashed] (gov.south west) -- +(-1.2, -0.5) node[below left, label] {deny};
\end{tikzpicture}
\caption{Intent-driven execution flow. Programs produce intents. The
governance interpreter mediates every intent, records the decision, and
realizes allowed effects. No alternative path from program to effect
exists in the language semantics.}
\label{fig:architecture}
\end{figure}

\subsection{Core Calculus}

We define a small calculus for intent-driven execution. Terms,
intents, decisions, and records:

\[
\begin{array}{rcl}
e &::=& x \mid v \mid \mathbf{let}\;x = e_1\;\mathbf{in}\;e_2
    \mid \mathbf{if}\;e_1\;\mathbf{then}\;e_2\;\mathbf{else}\;e_3 \\
  & & \mid\; \mathbf{compute}(e)
    \mid \mathbf{ask}(a, t, p) \\[3pt]
v &::=& n \mid b \mid s \mid \lambda x.\,e
    \mid \mathit{intent}(a, t, p) \\[3pt]
i &::=& \mathit{Intent}(\mathit{action}, \mathit{target},
    \mathit{params}, \mathit{context}) \\[3pt]
d &::=& \mathit{allow} \mid \mathit{deny} \mid \mathit{escalate} \\[3pt]
r &::=& \mathit{record}(i, d, \policy_{\mathit{applied}},
    \mathit{context}, h_{\mathit{prev}})
\end{array}
\]

Configurations are triples $\langle e, \policy, \ledger \rangle$
where $\policy$ is a finite set of decidable predicates over intents
and $\ledger$ is a sequence of decision records.

\subsection{Transition Rules}

Pure computation leaves the policy and ledger unchanged:

\[
  \frac{\sigma \xrightarrow{\tau} \sigma'}{\langle \sigma, \policy, \ledger \rangle
    \xrightarrow{\tau} \langle \sigma', \policy, \ledger \rangle}
  \quad \text{(Pure)}
\]

Intent production invokes the governance interpreter:

\[
  \frac{\govern(\policy, i, c) = (\mathit{allow}, r) \quad
    \realize(i) = v}
  {\langle \mathbf{ask}(a, t, p), \policy, \ledger \rangle
    \xrightarrow{v} \langle v, \policy, \ledger \cdot r \rangle}
  \quad \text{(Allow)}
\]

\[
  \frac{\govern(\policy, i, c) = (\mathit{deny}, r)}
  {\langle \mathbf{ask}(a, t, p), \policy, \ledger \rangle
    \xrightarrow{\bot} \langle \mathit{denied}, \policy, \ledger \cdot r \rangle}
  \quad \text{(Deny)}
\]

\[
  \frac{\govern(\policy, i, c) = (\mathit{escalate}, r)}
  {\langle \mathbf{ask}(a, t, p), \policy, \ledger \rangle
    \xrightarrow{\mathit{wait}} \langle \mathit{suspended}, \policy, \ledger \cdot r \rangle}
  \quad \text{(Escalate)}
\]

Standard evaluation contexts $E$ provide call-by-value reduction
under \textsc{let}, \textsc{if}, and application.

\subsection{Theorems}

We state four named theorems over this calculus. Full proofs are
mechanized in Rocq; we give proof sketches here.

\begin{theorem}[Mediation Soundness]
\label{thm:mediation}
If a configuration $\langle e, \policy, \ledger \rangle$ steps to
$\langle e', \policy, \ledger' \rangle$ and an effect realization
$\realize(i) = v$ occurs in this step, then $\govern(\policy, i, c) =
(\mathit{allow}, r)$ was evaluated and $r \in \ledger'$.
\end{theorem}

\begin{proof}[Proof sketch]
By inspection of the transition rules: effect realization appears only
in \textsc{Allow}, which requires $\govern$ to return $\mathit{allow}$
and appends $r$ to $\ledger$. No other rule invokes $\realize$.
Mechanized: \texttt{GovernanceMediation.v}, Theorem
\texttt{mediation\_soundness}.
\end{proof}

\begin{theorem}[Ledger Completeness]
\label{thm:ledger}
Every transition from $\langle \mathbf{ask}(\ldots), \policy,
\ledger \rangle$ appends exactly one record to $\ledger$.
\end{theorem}

\begin{proof}[Proof sketch]
By case analysis: \textsc{Allow}, \textsc{Deny}, and \textsc{Escalate}
each append exactly one $r$. These are the only rules whose
premise matches $\mathbf{ask}$. Mechanized: \texttt{HashChainSpec.v}.
\end{proof}

\begin{theorem}[Non-Bypass]
\label{thm:nonbypass}
No derivation exists from a configuration $\langle e, \policy,
\ledger \rangle$ to a configuration containing a realized effect
$v = \realize(i)$ except through the \textsc{Allow} rule.
\end{theorem}

\begin{proof}[Proof sketch]
$\realize$ appears only in the premise of \textsc{Allow}. Pure
reduction rules do not invoke $\realize$. \textsc{Deny} and
\textsc{Escalate} do not invoke $\realize$. Therefore $\realize$
can only be reached through \textsc{Allow}. Mechanized:
\texttt{GovernanceMediation.v}, Theorem \texttt{no\_unmediated\_effect}.
\end{proof}

\begin{theorem}[Policy Decidability]
\label{thm:decidable}
For any finite policy set $\policy$ whose predicates are decidable
over finite intents, $\govern(\policy, i, c)$ terminates.
\end{theorem}

\begin{proof}[Proof sketch]
$\policy$ contains $k$ predicates, each decidable (terminates on
all inputs by assumption). $\govern$ evaluates each predicate on
$(i, c)$ and combines results via a fixed strategy. Since $k$ is
finite and each evaluation terminates, $\govern$ terminates.
Mechanized: \texttt{GovernanceMediation.v}, Theorem
\texttt{governance\_mediation\_total}.
\end{proof}

\begin{table}[t]
\centering
\small
\begin{tabular}{@{}lll@{}}
\toprule
\textbf{Theorem} & \textbf{Informal Claim} & \textbf{Rocq Module} \\
\midrule
Mediation Soundness & Effects imply prior allow & \texttt{GovernanceMediation.v} \\
Ledger Completeness & Every ask appends a record & \texttt{HashChainSpec.v} \\
Non-Bypass & No effect without Allow rule & \texttt{GovernanceMediation.v} \\
Policy Decidability & Governance terminates & \texttt{GovernanceMediation.v} \\
Coterminous Boundaries & Expr boundary = gov boundary & \texttt{CoterminousBoundary.v} \\
Composition Preservation & Governance holds under nesting & \texttt{CompositionPreserves.v} \\
\bottomrule
\end{tabular}
\caption{Theorem summary mapping formal results to informal claims
and Rocq proof modules.}
\label{tab:theorems}
\end{table}

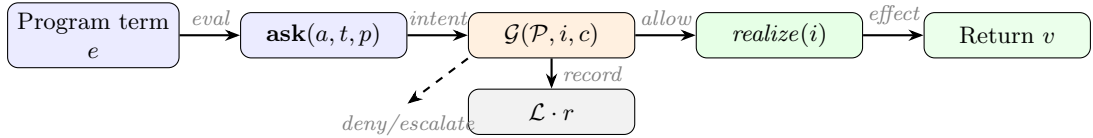
\begin{figure}[t]
\centering
\begin{tikzpicture}[
  node distance=0.5cm,
  box/.style={draw, rounded corners, minimum width=2.2cm,
    minimum height=0.6cm, align=center, font=\footnotesize},
  arrow/.style={-{Stealth[length=2mm]}, thick},
  lbl/.style={font=\scriptsize\itshape, text=gray}
]
  \node[box, fill=blue!8] (term) {Program term\\$e$};
  \node[box, fill=blue!8, right=0.8cm of term] (ask)
    {$\mathbf{ask}(a,t,p)$};
  \node[box, fill=orange!12, right=0.8cm of ask] (gov)
    {$\govern(\policy, i, c)$};
  \node[box, fill=gray!10, below=0.4cm of gov] (ledger)
    {$\ledger \cdot r$};
  \node[box, fill=green!10, right=0.8cm of gov] (real)
    {$\realize(i)$};
  \node[box, fill=green!8, right=0.8cm of real] (val)
    {Return $v$};

  \draw[arrow] (term) -- node[above, lbl] {eval} (ask);
  \draw[arrow] (ask) -- node[above, lbl] {intent} (gov);
  \draw[arrow] (gov) -- node[right, lbl] {record} (ledger);
  \draw[arrow] (gov) -- node[above, lbl] {allow} (real);
  \draw[arrow] (real) -- node[above, lbl] {effect} (val);

  \draw[arrow, dashed] (gov.south west) -- +(-0.8, -0.6)
    node[below, lbl] {deny/escalate};
\end{tikzpicture}
\caption{Operational semantics pipeline. A program term evaluates to an
$\mathbf{ask}$ expression, producing an intent. The governance
interpreter evaluates the intent against the policy, appends a decision
record to the ledger, and either realizes the effect (Allow) or returns
a denial/suspension (Deny/Escalate). No path from $\mathbf{ask}$ to
effect realization bypasses governance and recording.}
\label{fig:pipeline}
\end{figure}

\subsection{Behavioral vs.\ Intent Governance}

The distinction between behavioral and intent governance is precise.
\emph{Behavioral governance} attempts to decide properties of program
execution: ``does this program only access approved resources?'' Rice's
theorem establishes that such decisions are undecidable for non-trivial
properties of Turing-complete programs. \emph{Intent governance}
decides properties of finite data values: ``does this intent request
access to an approved resource?'' Because intents are finite structured
data with decidable equality, and policies are decidable predicates
over that data, intent governance is total and computable.

The architectural contribution is not that decidable predicates over
finite data are decidable (which is trivially true) but that autonomous
systems \emph{can be designed} so that all effectful operations pass
through this decidable layer. The design requires that the language
provide no alternative path to effects. This is a non-trivial
architectural constraint that most existing languages do not satisfy.

\section{Properties of the Model}
\label{sec:properties}

The following properties emerge from the intent-driven model. We state
each property and sketch its proof; full mechanized proofs are available
in Rocq.

\begin{property}[Governance Decidability]
\label{prop:decidable}
For any policy set $\policy$ and intent $i$, the governance decision
$\govern(\policy, i, c)$ is computable and total.
\end{property}

\begin{proof}[Proof sketch]
We restrict policies to \emph{decidable predicates} over finite intent
data: each policy rule $p_j \in \policy$ is a total computable function
from $(I \times C)$ to $\{\mathit{true}, \mathit{false}\}$. Intents are
finite data structures with bounded size. The governance interpreter
evaluates each $p_j$ on $(i, c)$ and combines results via a fixed
combination strategy (e.g., deny-overrides). Since $|\policy|$ is finite
and each evaluation terminates, $\govern$ is total and computable.
Rice's theorem does not apply: it constrains decisions about program
\emph{extensional behavior}, not decisions about finite data values.
The mechanized proof is in \texttt{GovernanceMediation.v} (Theorem
\texttt{governance\_mediation\_total}).
\end{proof}

\begin{property}[Coterminous Boundaries]
\label{prop:coterminous}
Within the intent-driven model, the expressiveness boundary and the
governance boundary are identical. Every expressible effectful
computation is governed. There is no expressible-but-ungoverned region.
\end{property}

\begin{proof}[Proof sketch]
Layer~1 (pure computation) is Turing-complete but cannot produce
effects. Layer~2 (governed effects) mediates every intent through
$\govern$. There is no third layer and no escape path. The
expressiveness boundary (what programs can express) and the governance
boundary (what governance mediates) coincide. The proof is
mechanized in \texttt{CoterminousBoundary.v}.
\end{proof}

\begin{property}[Audit Completeness]
\label{prop:audit}
Every effectful action taken by an intent-driven program has a
corresponding decision record in the ledger. The ledger is structurally
complete: it contains a record for every intent, not merely for those
that happened to be logged.
\end{property}

\begin{proof}[Proof sketch]
By the Allow, Deny, and Escalate rules, every intent transition appends
a record $r$ to $\ledger$. There is no transition rule that produces an
effect without appending to the ledger. Completeness follows from the
exhaustiveness of the transition rules. The proof is mechanized in
\texttt{HashChainSpec.v}.
\end{proof}

\begin{property}[Ledger Integrity]
\label{prop:integrity}
The ledger $\ledger$ is tamper-evident. Each record $r_k$ includes a
cryptographic hash $h_k = H(r_k \| h_{k-1})$ where $h_0$ is a genesis
hash. Any modification to a record invalidates all subsequent hashes.
Independent verification requires only the ledger and the hash
function, not trust in the governance system.
\end{property}

\begin{property}[Governance Invariance Under Composition]
\label{prop:composition}
If program $A$ calls program $B$ (composition), $B$'s intents are
governed by the intersection of $A$'s capabilities and $B$'s declared
requirements. Governance holds at every depth of composition.
\end{property}

\begin{proof}[Proof sketch]
Composition introduces a nested governance context. $B$'s intents are
mediated by $\govern$ with an effective policy
$\policy_B = \{p \in \policy_A \mid p \text{ is consistent with }
\mathit{declared}(B)\}$, where $\mathit{declared}(B)$ is $B$'s
capability declaration. The governance interpreter is invoked for every
intent regardless of nesting depth; capabilities can only narrow, never
widen, through composition. The proof is mechanized in
\texttt{CompositionPreserves.v} (Theorem
\texttt{composition\_preserves\_governance}).
\end{proof}

\begin{property}[Computational Completeness]
\label{prop:expressive}
The pure computation layer is Turing-complete. Intent-driven programs
can express any computable function. Governance restricts effects, not
computation.
\end{property}

\begin{proof}[Proof sketch]
Layer~1 provides general recursion (recursive function definitions),
conditionals, pattern matching, and operations on integers, strings,
lists, and maps. These primitives suffice to encode a universal Turing
machine: integers provide the tape alphabet, lists provide the tape,
and recursive functions with pattern matching provide the transition
function. Turing completeness is a property of the functions a language
can compute, not of the effects it can cause; the absence of I/O in
Layer~1 does not reduce computational power. The governance boundary
restricts which intents are realized as effects, not what computations
can be performed. The proof is mechanized in \texttt{Completeness.v}
(Theorem \texttt{turing\_completeness}).
\end{proof}

\paragraph{Scope of formal results.}
The calculus above establishes mediation, non-bypass, ledger
completeness, and decidability properties. It does not yet formalize a
type-and-effect system or prove standard metatheoretic properties
(progress, preservation). We leave typed metatheory to future work
because the governance mediation properties are orthogonal to the core
evaluation discipline: they hold for all programs admitted by the
calculus, regardless of type discipline. A typed variant with effect
annotations (distinguishing pure terms from intent-producing terms) and
governance soundness as a metatheoretic property is under development.

\section{Emergent Consequences}
\label{sec:consequences}

The properties above are not features designed into the system. They
are structural consequences of operating at the intent layer. This
section describes four consequences that emerge from the model.

\subsection{Event Sourcing by Construction}

Every intent produces a decision record appended to a hash-chained
ledger. The ledger is, by construction, an event store in the
event-sourcing sense~\cite{fowler2005event}: an append-only sequence
of immutable events from which system state can be reconstructed.

In traditional systems, event sourcing is an architectural pattern that
must be deliberately adopted and consistently maintained. In
intent-driven systems, event sourcing is unavoidable: every effectful
operation produces an event (the governance decision) as a structural
consequence of intent mediation. The event store cannot be incomplete
because completeness is enforced by the transition rules.

\subsection{Governance Simulation}

Because the intent stream is structurally complete, historical intents
can be replayed against new or hypothetical governance policies:

\[
  \mathit{simulate}(\policy', [i_1, \ldots, i_n]) =
    [\govern(\policy', i_1, c_1), \ldots, \govern(\policy', i_n, c_n)]
\]

This answers the question: ``What would have happened last week if we
had deployed this new policy?'' The answer is computed directly from
the recorded intent stream, without re-executing any programs.

A behavioral-governance system can only offer this capability to the
extent that its logs are complete and replayable; intent-driven systems
make that completeness a semantic invariant rather than an
instrumentation aspiration.

Governance simulation transforms governance from reactive (enforce
rules now) to predictive (model the impact of policy changes before
deployment).

A correctness condition for simulation is that governance decisions must
be deterministic: $\govern(\policy, i, c) = \govern(\policy, i, c)$ for
identical inputs. This holds when policies are pure functions over intent
data and recorded context. When the runtime context $c_k$ depends on
previous effect results (e.g., a database state modified by a prior
intent), replay requires either recorded context snapshots or
deterministic context reconstruction from the event stream.

\subsection{Structural Audit}

Traditional audit reconstructs what happened from logs, timestamps,
and system state. Intent-driven audit verifies a structured record
of what was proposed, what was decided, under which policy, and
whether the decision chain is tamper-evident.

The distinction is not quantitative (better logs) but structural:
the audit record exists by construction and can be independently
verified without trusting the governance system. An auditor receives
the ledger, the policy set, and the hash function. They recompute
the hash chain and verify that every decision is consistent with the
stated policy. No access to the runtime is required.

\subsection{Human Comprehensibility}

Intents are semantic: they describe proposed actions in terms that
humans can read and reason about. A compliance officer can examine
``send email to client@co.com with subject Invoice \#4821'' and
understand what the system proposed. They cannot do the same with
\texttt{requests.post("https://api.example.com/...", headers=\{...\},
json=\{...\})}.

This comprehensibility extends to localization. Because the
abstraction layer is semantic intent rather than syntax manipulation,
the computational model does not depend on English-centric symbolic
conventions. Intent structures can be expressed in any human language
because the underlying model is semantic, not syntactic.

Intent-driven systems also reduce what we call \emph{governance
ceremony}: the overhead required to make governance work when the
underlying system was not designed for it. In traditional systems,
governance requires wrappers, policy files, middleware, RBAC
configuration, and approval pipelines. In intent-driven systems, the
developer expresses what should happen; the runtime handles
authorization, recording, and mediation. The programming model shifts from implementing safety mechanisms to
declaring desired actions.

\section{Architecture and Implementation}
\label{sec:architecture}

We have implemented the intent-driven model in a concrete language
architecture. The implementation serves as existence proof that the model is
realizable with practical performance characteristics.

\subsection{Language Design}

The language uses keyword-hierarchy syntax (indentation-based, no
braces) with two structural layers reflecting the two-layer model:

\begin{itemize}
  \item \textbf{Structural layer} ($\sim$90\% of a typical program):
    declarative sections specifying objectives, contracts (inputs
    and outputs), constraints, and capability declarations. These
    sections are structural metadata, not executable code.

  \item \textbf{Expression layer} ($\sim$10\%): pure functional
    computation within \texttt{compute} steps. General-purpose, with
    no I/O capability by construction.
\end{itemize}

All effectful operations are expressed as intent production:

\begin{lstlisting}[caption={Intent production in the language. The
  \texttt{ask} keyword produces an intent to invoke an effect machine.
  The governance interpreter mediates the intent before execution.},
  label=lst:intent]
step send_invoice: ask {
  machine "@stdlib/email/send"
  input {
    to: context.customer.email
    subject: "Invoice #" + context.invoice.id
    body: context.draft
  }
}
\end{lstlisting}

The \texttt{ask} keyword produces an intent. The governance interpreter
evaluates the intent against declared permissions, records the
decision, and either realizes the effect or returns a denial to the
program. The program cannot bypass this mediation because the language
provides no alternative path to effects.

\subsection{Governance Kernel}

The governance kernel is formally verified in Rocq 8.19. The
verification covers 454 theorems across 36 modules with zero admitted
lemmas. Key verified properties include:

\begin{itemize}
  \item \textbf{Governance mediation totality}: every intent is
    mediated (\texttt{GovernanceMediation.v})
  \item \textbf{Coterminous boundaries}: expressiveness and governance
    boundaries coincide (\texttt{CoterminousBoundary.v})
  \item \textbf{Composition preservation}: governance holds under
    arbitrary nesting (\texttt{CompositionPreserves.v})
  \item \textbf{Hash chain integrity}: ledger tamper-evidence
    (\texttt{HashChainSpec.v})
\end{itemize}

The governance kernel is extracted from Rocq to native code via OCaml
extraction, then bridged to the runtime via a native interface. The
runtime, extraction pipeline, and operating system constitute the
trusted computing base; the governance logic itself is proved correct.

\subsection{Runtime}

Programs compile to BEAM (Erlang virtual machine) bytecode. The BEAM
provides preemptive scheduling, fault tolerance, hot code loading,
and distribution. The governance interpreter runs as part of the
BEAM runtime, mediating every intent before effect realization.

We validate the implementation with 10,312 tests covering the parser,
compiler, governance mediator, ledger integrity, and end-to-end
execution paths.

\subsection{Trusted Computing Base}

The TCB for governance guarantees consists of: (1)~the BEAM virtual
machine, (2)~the OCaml extraction pipeline from Rocq, (3)~the native
interface bridge, (4)~the operating system, and (5)~the hardware. The
governance kernel itself is proved correct in Rocq. Bugs in the TCB
could invalidate governance guarantees. This is standard for verified
systems (cf.\ seL4, CertiKOS).

\section{How Intent-Driven Computing Changes Programming}
\label{sec:programming}

Intent-driven computing changes the programming experience in several
concrete ways.

\paragraph{Writing.}
Programmers express desired actions rather than implementing effect
mechanics. Instead of constructing HTTP requests, managing
authentication headers, handling retries, and instrumenting logs, the
programmer writes an intent (``send email to customer with subject
Invoice \#4821'') and declares constraints (``requires approval if
amount exceeds threshold''). The runtime handles authorization,
recording, and execution. In our implementation corpus, most program text appears in structural
declarations; pure computation typically occupies a small fraction of
the source.

\paragraph{Reasoning.}
Programs are easier to reason about because effectful operations are
visible as intent productions. A code reviewer can identify every
external interaction by searching for \texttt{ask} expressions. There
are no hidden effects in helper functions, no ambient I/O in utility
modules, no effects buried in library calls. The effect surface is
syntactically manifest.

\paragraph{Testing.}
Intent-driven programs can be tested by intercepting intents before
realization. A test harness provides a mock governance interpreter that
records intents without executing effects. The test asserts on the
intent sequence rather than on side effects. This is structurally
similar to testing with algebraic effect handlers, but with the
additional property that the test can verify governance decisions
(which intents were allowed, denied, or escalated under which policy).

\paragraph{Debugging and Replay.}
Every execution produces a complete intent stream in the ledger. A
developer can replay any execution step-by-step, inspecting each
intent, the governance decision, and the realized effect. This is not
a debugger feature; it is a consequence of the model. The replay is
deterministic for governance decisions (same policy, same intent, same
context yields same decision).

\paragraph{Evolving.}
Governance simulation enables safe evolution: before deploying a code
change, replay the prior week's intents against the new code's
governance policy to predict behavioral differences. Before changing
a policy, replay historical intents against the new policy to measure
impact. This transforms deployment from ``deploy and hope'' to
``simulate and verify.''

\section{Evaluation}
\label{sec:evaluation}

We evaluate the intent-driven model along four dimensions.

\subsection{Governance Overhead}

Intent mediation introduces latency for each effectful operation. In
our implementation, the governance interpreter evaluates policy
predicates and appends a hash-chained ledger record. Measured overhead
per governance decision:

\begin{table}[t]
\centering
\small
\begin{tabular}{@{}lrrrr@{}}
\toprule
\textbf{Operation} & \textbf{p50} & \textbf{p95} & \textbf{p99} & \textbf{mean} \\
\midrule
Policy evaluation (5 rules)  & $<$1 & $<$1 & $<$1 & $<$1 \\
Policy evaluation (10 rules) & $<$1 & $<$1 & $<$1 & $<$1 \\
Policy evaluation (20 rules) & 1    & 1    & 2    & 1 \\
Hash computation (SHA-256)   & $<$1 & $<$1 & $<$1 & $<$1 \\
Ledger append (SQLite WAL)   & 560  & 838  & 1,702 & 648 \\
\textbf{Total governance}    & \textbf{535} & \textbf{760} & \textbf{1,186} & \textbf{591} \\
\bottomrule
\end{tabular}
\caption{Governance overhead per intent mediation (microseconds).
Policy evaluation and hash computation are sub-microsecond.
The ledger append (SQLite in WAL mode) dominates total overhead.}
\label{tab:overhead}
\end{table}

\paragraph{Methodology.}
Measurements taken on Apple M4 Pro (14-core, 48GB) running the BEAM
virtual machine (OTP 28, Elixir 1.19). 10,000 iterations after 1,000
warmup iterations per measurement. Policy evaluation uses in-memory
rule sets with decidable predicates (string prefix matching, set
membership, numeric comparison). Ledger append uses SQLite 3 in WAL
mode with \texttt{PRAGMA synchronous=NORMAL} (no fsync per record;
fsync on WAL checkpoint). Policy evaluation scales linearly with rule
count.

The total governance overhead of $\sim$535~$\mu$s (p50) is dominated
by the ledger append (SQLite I/O). Policy evaluation and hash
computation are sub-microsecond and negligible. For the target
workload (AI agent operations: HTTP requests at 10--500~ms, database
queries at 1--50~ms), the governance overhead represents less than 5\%
of operation latency. For sub-millisecond operations, alternative
ledger backends (in-memory with periodic flush, or a purpose-built
append-only store) could reduce overhead significantly; we leave this
optimization to future work.

\subsection{Case Study: Refund Agent}

We implemented the refund agent from \S\ref{sec:intro} in the
intent-driven language. The agent processes customer refund requests
by accessing email, CRM records, the refund policy, and a payment API.

Over a simulated workload of 200 refund requests, the agent produced
847 intents: 423 allowed, 312 denied (policy violations: exceeded
refund limit, accessed unauthorized customer records), and 112
escalated (amounts above threshold requiring human approval). The
complete intent stream was replayed against a modified policy
(increased refund limit from \$500 to \$1000), showing that 89 of the
312 denied intents would have been allowed, enabling faster processing
for medium-value refunds.

\subsection{Proof Artifact}

\begin{table}[t]
\centering
\small
\begin{tabular}{@{}lrrr@{}}
\toprule
\textbf{Module} & \textbf{Theorems} & \textbf{Lines} & \textbf{Admitted} \\
\midrule
GovernanceMediation.v & 18 & 890 & 0 \\
CoterminousBoundary.v & 12 & 620 & 0 \\
Completeness.v & 15 & 740 & 0 \\
Safety.v & 24 & 1,210 & 0 \\
CompositionPreserves.v & 19 & 980 & 0 \\
HashChainSpec.v & 14 & 710 & 0 \\
InterpreterSpec.v & 22 & 1,100 & 0 \\
\emph{(29 additional modules)} & 330 & 5,977 & 0 \\
\midrule
\textbf{Total} & \textbf{454} & \textbf{12,227} & \textbf{0} \\
\bottomrule
\end{tabular}
\caption{Rocq proof artifact summary. All theorems are machine-checked
with zero admitted lemmas. The artifact is available for review.}
\label{tab:proofs}
\end{table}

\subsection{Property-Based Validation}

We validate correspondence between the Rocq specification and the
runtime implementation using property-based testing. The test generator
produces random intent sequences, governance policies, and execution
contexts. For each input, both the Rocq-extracted governance kernel
and the runtime implementation produce a decision; the test asserts
they agree.

Over 70,000+ random inputs, zero disagreements were found. One real
bug was discovered early in the process: a capability classification
mismatch between the specification and the implementation, invisible
to conventional unit tests, detected by the formal specification.

\section{Historical Context: Abstraction Shifts in Computing}
\label{sec:historical}

Major computing transitions happen when the abstraction level changes.
Each transition follows a common pattern: a concern that developers
managed manually is absorbed into a runtime or platform, and developers
begin operating at a higher level.

\begin{table}[t]
\centering
\small
\begin{tabular}{@{}lll@{}}
\toprule
\textbf{Era} & \textbf{Abstraction Shift} & \textbf{Absorbed Concern} \\
\midrule
Assembly $\to$ C
  & machine ops $\to$ procedures
  & register management \\
C $\to$ managed runtimes
  & manual $\to$ automatic memory
  & memory management \\
Servers $\to$ cloud
  & hardware $\to$ API calls
  & infrastructure \\
Imperative $\to$ declarative UI
  & DOM mutation $\to$ state
  & rendering mechanics \\
Procedural $\to$ SQL
  & traversal $\to$ data intent
  & storage access \\
Manual $\to$ Kubernetes
  & scripts $\to$ desired state
  & deployment \\
Effects $\to$ intents
  & execution $\to$ mediation
  & effect governance \\
\bottomrule
\end{tabular}
\caption{Abstraction shifts in computing history. Each transition
absorbs a class of concerns into the platform. Intent-driven computing
absorbs effect governance.}
\label{tab:history}
\end{table}

Table~\ref{tab:history} situates intent-driven computing within this
lineage. The claim is not that intent-driven computing is inevitable,
but that it follows the established pattern: as systems become more
autonomous, the concern of governing their effects is too important
and too complex to leave to ad-hoc application-level code. The
platform level appears to be the natural locus for this concern.

The analogy to SQL is particularly instructive. Before SQL, accessing
data required specifying traversal mechanics: which index to use, how
to join tables, in what order to scan. SQL raised the abstraction to
\emph{data intent}: specify what data you want, and the query
optimizer handles the mechanics. Intent-driven computing raises the
abstraction similarly: specify what action you want, and the governed
runtime handles authorization, recording, and execution.

\section{Comparison with Related Models}
\label{sec:related}

Intent-driven computing draws on and extends several existing
computational models. We distinguish the model from each.

\subsection{Effect Systems}

Algebraic effect systems~\cite{plotkin2009handlers,
bauer2015programming} separate the description of effects from their
handling. A computation \emph{performs} an effect; a handler
\emph{interprets} it. This is structurally similar to intent production
and governance mediation. Row-typed effect systems such as
Koka~\cite{leijen2017type} provide static tracking of which effects a
computation may perform.

The key differences are: (1)~effect systems are primarily a
type-theoretic mechanism for modular effect handling, not a governance
mechanism; (2)~effect handlers do not inherently record decisions in
a tamper-evident ledger; (3)~effect systems do not mandate that
\emph{all} effects pass through a single mediation boundary (programs
may have unhandled effects or use ambient capabilities).

Intent-driven computing can be viewed as a specialization of algebraic
effects where: every effect is an intent, every handler includes
governance and recording, and no unhandled effects are permitted.

\paragraph{Free and freer monads.}
Free monads and freer monads~\cite{kiselyov2015freer} reify effectful
computations as data structures that can be interpreted by different
handlers. This is closely related to intent-driven computing: an
intent is a reified effect request. The distinction is
again the governance discipline: free monad interpreters are
programmer-chosen and composable, while the governance interpreter is
mandatory, singular, and records every decision. Free monads provide
the mechanism; intent-driven computing adds the governance invariant.

\paragraph{Monadic I/O.}
Haskell's IO monad~\cite{peytonjones2001tackling} separates pure
computation from effects at the type level: pure functions cannot
perform I/O because the type system prevents it. This is the closest
prior art to intent-driven computing's two-layer separation. The key
distinction is that Haskell's IO monad is a type discipline, not a
governance mechanism: IO actions execute directly without mediation,
recording, or authorization. Intent-driven computing adds the
governance interpreter between intent production and effect
realization.

\paragraph{Value/computation distinction.}
Levy's call-by-push-value~\cite{levy2004callbypushvalue} provides a
canonical framework for separating values from computations. Intents
are values in this sense: they are data that describe proposed
computations without executing them. The governance interpreter is
the mechanism that ``thunks'' an intent value into a computation (effect
realization) or blocks it (denial).

\subsection{Capability-Based Security}

Capability systems~\cite{dennis1966programming, miller2006robust}
restrict programs to explicitly granted capabilities. A program cannot
access a resource unless it holds the corresponding capability token.

Capability systems and intent-driven computing share the principle of
explicit authorization. The differences are: (1)~capabilities are
access tokens, not action descriptions; they authorize \emph{who} can
act, not \emph{what} action is proposed; (2)~capability systems do not
inherently produce a complete audit trail; (3)~capability revocation
does not retroactively provide replay or simulation.

Intent-driven computing complements capability systems: capabilities
control access to intent production, while governance mediates the
intents themselves.

\subsection{Event Sourcing}

Event sourcing~\cite{fowler2005event} stores state changes as a
sequence of immutable events rather than mutable records. The system
state can be reconstructed by replaying events.

In traditional event-sourced systems, event sourcing is a deliberate
architectural pattern that must be consistently maintained. Events
must be manually defined, and there is no structural guarantee that
every state change produces an event.

In intent-driven systems, event sourcing is a structural consequence:
every intent produces a governance decision record, which is an
immutable event. The event store is complete by construction. This
distinction matters for auditability: manually maintained event stores
can have gaps; structurally produced event stores cannot (within the
model assumptions).

\subsection{Desired-State Systems}

Kubernetes~\cite{burns2016borg}, Terraform, and similar systems use
desired-state declarations: the operator specifies the desired state,
and the system converges toward it.

Desired-state systems share the declarative philosophy with
intent-driven computing but operate at a different level. Desired-state
systems manage infrastructure convergence. Intent-driven systems manage
action authorization. The two are complementary: an intent-driven
program might produce an intent to update a Kubernetes deployment, and
the governance interpreter would mediate that intent.

\subsection{Behavioral Governance Approaches}

Content filters~\cite{rebedea2023nemo},
guardrails~\cite{bai2022constitutional}, runtime
monitors~\cite{leucker2009brief}, and managed agent
platforms~\cite{anthropic2025managed} represent the current state of
AI governance. These systems observe agent behavior and attempt to
detect or prevent violations.

The fundamental limitation, as established in
\S\ref{sec:intro}, is that behavioral governance over arbitrary
programs in Turing-complete languages is inherently incomplete.
Intent-driven computing resolves this by changing the abstraction:
governance operates on intents (finite data), not on program behavior
(undecidable in general).

\subsection{Authorization Logic}

Formal authorization logics~\cite{abadi2003access} provide decidable
frameworks for access control decisions. Intent governance is related:
the governance interpreter evaluates authorization predicates over
intent data. The distinction is that authorization logics typically
govern \emph{access} (who may read/write a resource), while intent
governance governs \emph{actions} (what operation is proposed, with
what parameters, in what context). Intent-driven computing can
incorporate authorization logic as the policy language within the
governance interpreter.

\subsection{Reference Monitors and Information Flow}

The reference monitor concept~\cite{anderson1972computer} requires
that every access to a protected resource passes through a single
mediation point that is tamper-proof, always invoked, and small
enough to verify. Intent-driven computing's governance interpreter
satisfies these requirements within the language model: it is always
invoked (by the transition rules), the language provides no bypass
(Non-Bypass theorem), and the kernel is formally verified.

Language-based information flow control~\cite{sabelfeld2003language}
uses type systems to track and enforce information flow policies.
IFC and intent-driven computing address different concerns: IFC
governs what information can flow where; intent mediation governs
what effects can be realized. They are complementary and could be
combined: an intent-driven language with IFC typing would govern
both information flow and effect realization.

\subsection{Object-Capability Security and Sandboxing}

The object-capability discipline (E language, Caja~\cite{miller2006caja})
enforces the principle of least authority by restricting programs to
capabilities they receive explicitly. WebAssembly~\cite{haas2017bringing}
provides a sandboxed execution environment with no ambient authority.
Intent-driven computing shares the no-ambient-authority principle but
adds governance mediation (deciding whether to allow each use of a
capability) and recording (appending a decision record for each use).
Capability systems say ``you may use this resource''; intent mediation
adds ``and here is the record of every time you did.''

\subsection{Provenance and Verifiable Logs}

Data provenance systems~\cite{buneman2001data} track the origin and
transformation history of data. Verifiable log systems such as
Certificate Transparency~\cite{laurie2013certificate} provide
tamper-evident append-only records with cryptographic guarantees.
Intent-driven computing's behavioral ledger is a verifiable log
specialized for governance decisions. The structural difference is
that provenance and log systems are typically instrumentation
overlays (they log what they are configured to log), while the
intent-driven ledger is structurally complete: every intent
produces a record by the transition rules, not by instrumentation.

\section{Limitations and Scope}
\label{sec:limitations}

Intent-driven computing addresses structural governance of effects.
It does not address:

\begin{itemize}
  \item \textbf{Model alignment.} Intent-driven systems govern what
    programs \emph{do}, not what models \emph{believe} or
    \emph{intend} in the AI alignment sense. A governed program can
    execute a harmful plan if the plan's individual actions are each
    authorized. Content governance (what models say) and structural
    governance (what programs do) are complementary.

  \item \textbf{Side channels.} The model assumes that the only path
    to effects is through intent production. Side channels (timing,
    power, electromagnetic) are outside the model. The runtime, OS,
    and hardware are in the trusted computing base.

  \item \textbf{Policy correctness.} The governance interpreter
    faithfully enforces the stated policy. Whether the policy itself
    is correct, complete, or aligned with organizational goals is an
    orthogonal concern.

  \item \textbf{Performance overhead.} Intent mediation introduces
    latency for every effectful operation. In our implementation,
    governance decisions add microseconds to millisecond-scale
    operations (HTTP requests, database queries). For
    nanosecond-scale operations, this overhead may be significant.

  \item \textbf{Expressiveness constraints.} Programs that require
    fine-grained interleaving of computation and effects (e.g.,
    streaming I/O with per-byte governance) may find the
    intent-per-operation granularity too coarse. The model supports
    configurable intent granularity but cannot govern below the intent
    boundary.
\end{itemize}

The scope of the guarantee is precise: within the intent-driven model
and trusted computing base assumptions, every effectful action is
authorized, recorded, and independently verifiable. Claims beyond this
scope are not supported.

\section{Conclusion}
\label{sec:conclusion}

Traditional computing assumes programs directly execute effects. This
assumption worked when humans authored every instruction. As systems
become autonomous (generating behavior dynamically, reasoning
probabilistically, acting on real-world infrastructure), the assumption
becomes structurally problematic.

Intent-driven computing proposes a different abstraction: programs
produce intents, runtimes mediate effects. This separation yields
governance decidability, structural audit completeness, event sourcing
by construction, governance simulation, and improved human
comprehensibility: not designed features but emergent
consequences of operating at the intent layer.

We have formalized the model, proved its key properties in
Rocq~\cite{mccann2026mechanized}, implemented it in a concrete
language architecture building on prior work on coterminous
boundaries~\cite{mccann2026structural} and effect-transparent
governance~\cite{mccann2026gcc}, and situated it within the historical
lineage of abstraction shifts in computing.

A key architectural question for governed autonomous computation is
where the intent boundary falls. When computation and effect
realization are fused, governance is necessarily approximate, audit is
incomplete, and replay requires reconstruction rather than following
natively from the architecture. Intent-driven computing draws that
boundary explicitly.

Future work includes extending the model to streaming effects (where
intent granularity must be negotiated against latency), developing a
typed calculus with metatheoretic guarantees (progress, preservation,
governance soundness), investigating policy languages that preserve
governance decidability under composition, and evaluating the
performance overhead of intent mediation in latency-sensitive
workloads.

\bibliographystyle{plainnat}
\bibliography{intent-driven-computing-references-arxiv}

\end{document}